\documentclass[aps,prd,twocolumn,nofootinbib]{revtex4-2}
\usepackage[T1]{fontenc}
\usepackage{lmodern}
\usepackage{microtype}
\usepackage{bm}
\usepackage{graphicx}
\usepackage{hyperref}
\usepackage{amsthm}
\usepackage{xcolor}
\usepackage{amsmath}
\definecolor{linkblack}{RGB}{20,20,20}
\definecolor{citeblue}{RGB}{0,70,140}
\definecolor{urlblue}{RGB}{0,90,120}
\usepackage{float}
\hypersetup{
  colorlinks=true,
  linkcolor=linkblack,
  citecolor=citeblue,
  urlcolor=urlblue
}

\newtheorem{theorem}{Theorem}[section]
\newtheorem{corollary}[theorem]{Corollary}

\begin{document}

\title{Black Hole Interiors as a Laboratory for Time-Dependent Classical Double Copy}

\author{Damien A. Easson}
\email{easson@asu.edu}
\affiliation{
Department of Physics \& Beyond Center for Fundamental Concepts in Science,
Arizona State University, Tempe, Arizona 85287-1504, USA}

\author{Tucker Manton}
\email{tucker_manton@ucas.ac.cn}
\affiliation{School of Fundamental Physics and Mathematical Sciences, Hangzhou Institute for Advanced Study, University of Chinese Academy of Sciences (UCAS-HIAS), Hangzhou 310024, China}

\begin{abstract}
The classical double copy relates gravity and gauge theory, but its most explicit realizations remain concentrated in stationary or highly symmetric settings. We show that trapped regions of black-hole
geometries provide an exact setting for time-dependent classical double copy.
More generally, a class of Kantowski--Sachs cosmologies admits an intrinsic double-copy construction: the Kerr--Schild scalar and
single-copy field are uniquely reconstructible from the cosmological metric
alone, without prior knowledge that the spacetime is a black-hole interior or
reference to any exterior completion. The reconstructed field admits both the
standard flat-background Kerr--Schild realization and a complementary
Maxwell realization on the Kantowski--Sachs geometry, where it describes a
homogeneous, time-dependent electric field. We prove that this class is characterized by a relation between the Kantowski--Sachs scale factors equivalent to
$p_{\parallel}=-\rho$. For Schwarzschild, the single-copy electric field diverges along the interior evolution, whereas for the regular Bardeen solution it remains finite throughout the trapped region and extends smoothly into the regular static core.
The Bardeen geometry violates the strong energy condition in a compact
central region, whereas the corresponding single-copy Maxwell field remains
regular and satisfies the standard classical energy conditions. We further
show that the Bardeen horizon phase structure is encoded in the single-copy
scalar. Thus trapped Kerr--Schild interiors provide an exact local realization of time-dependent classical double copy.
\end{abstract}

\maketitle

\section{Introduction}
\label{sec:intro}

The classical double copy provides a remarkable correspondence between gauge
theory and gravity, with roots in amplitude relations and later extensions to
exact classical solutions
\cite{Kawai:1985xq,Bern:2008qj,Bern:2010ue,White:2018doublecopy,Monteiro:2014cda}.
In its original Kerr--Schild formulation, the map applies most directly to
stationary black holes, plane waves, and related highly symmetric geometries,
where a gravitational metric written as
\begin{equation}\label{eq:ksmet}
 g_{\mu\nu}=\eta_{\mu\nu}+2\phi\,k_\mu k_\nu
\end{equation}
is associated with the Abelian gauge field
\begin{equation}
 A_\mu=\phi\,k_\mu.
\end{equation}
This framework has been highly successful for static black holes and their
matter-supported generalizations
\cite{Monteiro:2014cda,Luna:2015paa,Ridgway:2015fdl,Easson:2020nsbh}, while
explicit treatments of genuinely time-dependent exact spacetimes remain
comparatively scarce
\cite{Luna:2016hge,CarrilloGonzalez:2018eex,Ortaggio:2023ksk,CarrilloGonzalez:2024sdc,Ilderton:2024arb,Kent:2025pvu}.
Related exact-solution and source-structure developments have broadened
the classical and Weyl double-copy program
\cite{Easson:2023iso,Easson:2021srw,Easson:2023emw,Easson:2026vep}, which is a mapping between curvature quantities in the spinor formalism \cite{Luna:2018dpt}.

Most progress on time-dependent classical double copy has focused on highly
symmetric exact backgrounds or perturbative constructions. This raises the question of whether an exact construction can be obtained in a spacetime that is anisotropic, strongly time dependent, and directly tied to black-hole causal structure. 

Here we show trapped black-hole interiors provide such a
setting. Once the horizon is crossed, the areal radius becomes timelike and the
geometry is naturally interpreted as a cosmological spacetime. In particular, the
Schwarzschild interior is described by a Kantowski--Sachs metric: a homogeneous but
anisotropic cosmology with two distinct scale factors, one associated with the
longitudinal direction and the other with the two-sphere
\cite{Kantowski:1966,Wald:1984gr,Doran:2008}.~\footnote{Kantowski--Sachs spacetimes are locally rotationally
symmetric and therefore lie in the Petrov-\(D\) or conformally flat branch.
This suggests that a wider class of time-dependent Kantowski--Sachs
geometries may admit a Weyl-double-copy description. For matter-supported
cases, one would require a sourced Weyl-double-copy~\cite{Easson:2021srw}.} The trapped geometry itself contains sufficient information to reconstruct
the Kerr--Schild scalar and associated gauge field without reference to an
exterior completion. The standard single copy is defined on the flat
background and remains static there. Below we show that a complementary realization
on the full Kantowski--Sachs geometry makes the cosmological time dependence
manifest.

More speculative proposals have
suggested that black-hole interiors may seed new universes, especially when
the singularity is replaced by a regular de Sitter-like core
\cite{Pathria:1972,Easson:2001qf,Frolov:1989,Frolov:1990,Dymnikova:2019,Easson:2026juf}.
But here we use the interior cosmological viewpoint to formulate the classical double copy directly within time-dependent black-hole interiors.

The interior cosmological interpretation remains fully compatible with Birkhoff's theorem \cite{Birkhoff1923,Jebsen1921,Alexandrow1923,Eiesland1921,Eiesland1925,Deser:2004gi,Schleich:2009uj,Easson:2023ytf,Easson:2026cco}.
Spherical symmetry in vacuum still fixes the local geometry to be
Schwarzschild; there are no genuinely new spherically symmetric vacuum
dynamics. The causal character
of the coordinates changes in the trapped region. The
Kantowski--Sachs form of the Schwarzschild interior reflects a
region-adapted causal reinterpretation of the same local vacuum solution,
and not a violation of Birkhoff rigidity.

Regular black holes provide an even richer setting, replacing the
Schwarzschild singularity by a nonsingular core while remaining within an exact
Kerr--Schild class on the gravity side
\cite{Bardeen:1968,AyonBeato:2000,Ansoldi:2008}. For $M/g\gg1$, the outer horizon lies at $r_+\simeq 2M$, while the inner horizon is parametrically smaller than the regularization scale $g$,
\begin{equation}
r_-\simeq \sqrt{\frac{g^3}{2M}}
=\frac{g}{\sqrt{2M/g}} .
\end{equation}
Thus the nonstatic interior band $r_-<r<r_+$ is parametrically large in areal radius. Here $g$ is the parameter that regulates the would-be curvature singularity at $r=0$; it may be interpreted as a zero-point length \cite{Nicolini:2019irw} or, in the nonlinear-electrodynamics construction, as a magnetic monopole charge \cite{AyonBeato:2000}.

We further show that this class admits an intrinsic characterization. Among
Kantowski--Sachs cosmologies, those arising from trapped Kerr--Schild interiors
are singled out by a distinguished relation between the scale factors,
equivalently by the longitudinal condition \(p_\parallel=-\rho\). Schwarzschild and Bardeen provide the singular and regular examples, respectively.

The remainder of this paper is organized as follows. In
Sec.~\ref{sec:KSvac} we review the Kantowski--Sachs form of the Schwarzschild
and Bardeen interiors. In Sec.~\ref{sec:staticdc} we summarize the underlying
static Kerr--Schild single copies. In Sec.~\ref{sec:Bardeen} we derive the
time-dependent single-copy description on trapped Kantowski--Sachs patches and
show how it is characterized and reconstructed from interior cosmological
data. In Sec.~\ref{sec:energy} we compare the gravity- and gauge-side energy
conditions. In Sec.~\ref{sec:horizons} we show how the Bardeen horizon
structure is encoded in the single-copy scalar. We conclude in
Sec.~\ref{sec:concl} with a discussion of implications and future directions.
Throughout, we use the mostly-plus metric signature and units in which
\(G=c=1\).

\section{Spherically symmetric black-hole interiors}
\label{sec:KSvac}

Inside the Schwarzschild black hole event horizon (\(0<r<2M\)), the
coordinate roles swap: the Schwarzschild time \(t\) becomes spacelike while the
areal radius \(r\) is timelike. For the one-function spherically symmetric class, this causal interchange is most apparent in Schwarzschild-like coordinates,
\begin{equation}\label{spheresym}
    ds^2=-f(r)dt^2+f(r)^{-1}dr^2+r^2d\Omega_2^2,
\end{equation}
where,
\begin{equation}
 f(r)=1-\frac{2M}{r}.
\end{equation}
Inside the horizon, $f(r)<0$. Defining $T\equiv r$, $\chi\equiv t$, the metric becomes 
\begin{equation}\label{interiorregion}
    ds^2=-|f(T)|^{-1}dT^2+|f(T)|d\chi^2+T^2d\Omega_2^2.
\end{equation}
By introducing the proper time
\begin{equation}\label{propertime}
    \frac{d\tau}{dT}=|f(T)|^{-1/2},
\end{equation}
the metric becomes
\begin{equation}
 ds^2=-d\tau^2+a(\tau)^2d\chi^2+b(\tau)^2d\Omega_2^2,
\label{eq:genKSform}
\end{equation}
where the profile functions are 
\begin{equation}\label{KSprofilefunctions}
    a(\tau)=\sqrt{|f(T(\tau))|}, \ \ \ \ b(\tau)=T(\tau).
\end{equation}
This geometry is known as the Kantowski-Sachs cosmological metric \cite{Kantowski:1966}, and represents an anisotropic universe with two scale factors, $a(\tau)$ and $ b(\tau)$, characterizing expansion/contraction in the longitudinal and transverse directions, respectively. 

We will focus on the quintessential black hole examples given by the Schwarzschild and Bardeen solutions. However, the simple form of \eqref{spheresym} encompasses a large class of metrics that are single or multi-horizon. These include the Reissner-Nordstr\"om, Dymnikova \cite{Dymnikova:2019}, Hayward \cite{Hayward:2005gi}, and Bronnikov solutions \cite{Bronnikov:2000vy}, solutions from non-commutative geometry \cite{Nicolini:2008aj}, among others (see 
\cite{Stephani:2003tm} for a comprehensive review of solutions in general relativity and \cite{Ansoldi:2008jw,Bambi:2023try} for reviews on regular black hole solutions). Each of these solutions can accommodate a cosmological constant of either sign.  Thus, our Kantowski--Sachs construction applies to a broad class of spherically symmetric black-hole solutions having \(g_{tt}g_{rr}=-1\) in areal-radius coordinates.

\textit{Schwarzschild solution.} For Schwarzschild, defining a timelike coordinate \(T\equiv r\) and
a spacelike coordinate \(\chi\equiv t\), one obtains the interior line element
\begin{equation}
 ds^2=-\frac{dT^2}{\frac{2M}{T}-1}+\Bigl(\frac{2M}{T}-1\Bigr)d\chi^2+T^2d\Omega_2^2,
\label{eq:InteriorSchwarzschild}
\end{equation}
where \(d\Omega_2^2=d\theta^2+\sin^2\theta\,d\varphi^2\). This interior metric
can be recast in the form of \eqref{eq:genKSform} by introducing a proper time \(\tau\)
defined via
\begin{equation}
 \frac{d\tau}{dT}=\sqrt{\frac{T}{2M-T}}.
\label{eq:SchProperTime}
\end{equation}
The scale factors in the longitudinal and transverse directions are 
\begin{equation}
 a(\tau)=\sqrt{\frac{2M}{T(\tau)}-1},
 \qquad
 b(\tau)=T(\tau).
\label{eq:KSfactors}
\end{equation}
A natural interior observer is one comoving with the Kantowski--Sachs slicing,
i.e.\ at fixed \((\chi,\theta,\varphi)\). Such an observer is not static in the
exterior Schwarzschild sense, but evolves along the proper time \(\tau\)
while experiencing a genuinely time-dependent anisotropic geometry.
For Schwarzschild the proper-time relation can be integrated exactly by
introducing a parameter \(\eta\in(0,\pi/2)\):
\begin{equation}
 T(\eta)=2M\sin^2\eta,
 \qquad
 \tau(\eta)=2M\left(\eta-\sin\eta\cos\eta\right),
\label{eq:etaParam}
\end{equation}
where we have fixed the additive constant so that \(\tau=0\) at \(T=0\).
In terms of \(\eta\), the scale factors are
\begin{equation}
 a(\eta)=\cot\eta,
 \qquad
 b(\eta)=2M\sin^2\eta.
\label{eq:etaScales}
\end{equation}

The evolution is highly anisotropic: as \(T\to0\), \(b(\tau)\to0\) while
\(a(\tau)\to\infty\). Equivalently, the anisotropy is encoded in the
directional Hubble parameters
\(
H_{\parallel}\equiv \dot a/a
\)
and
\(
H_{\perp}\equiv \dot b/b
\),
where the dot denotes \(d/d\tau\). The associated expansion and shear scalars are
\begin{equation}
\Theta = H_{\parallel}+2H_{\perp},
\qquad
\sigma^2=\frac13(H_{\parallel}-H_{\perp})^2,
\label{eq:expShear}
\end{equation}
with \(\sigma^2=\tfrac{1}{2} \sigma_{ab}\sigma^{ab}\).
For the Schwarzschild interior,
\begin{equation}
H_{\parallel}
=-\frac{M}{T^{3/2}(2M-T)^{1/2}},
\qquad
H_{\perp}
=\frac{(2M-T)^{1/2}}{T^{3/2}},
\label{eq:HubbleSch}
\end{equation}
so \(H_{\parallel}\neq H_{\perp}\). Hence \(\sigma^2\neq0\), and in fact the
shear diverges as \(T\to0\).

For the natural geodesic, irrotational Kantowski--Sachs congruence, this
anisotropy is dynamically sourced by the electric part of the Weyl tensor,
\begin{equation}
E_{ab}=C_{acbd}u^c u^d,
\label{eq:ElectricWeyl}
\end{equation}
whose orthonormal-frame eigenvalues scale as \(M/T^3\). For the same
congruence, the Raychaudhuri equation reduces in vacuum to
\begin{equation}
\dot{\Theta}
=
-\frac{1}{3}\Theta^{2}
-2\sigma^{2},
\label{eq:RaychaudhuriVac}
\end{equation}
since \(R_{ab}u^{a}u^{b}=0\). Thus there is no Ricci-driven focusing: the
additional focusing beyond the isotropic \(-\Theta^{2}/3\) term is entirely
due to the shear. Because the shear is itself driven by \(E_{ab}\), the
singular focusing may be viewed as indirectly Weyl-driven: the Schwarzschild
tidal field generates anisotropy, and that anisotropy in turn drives the
congruence to focus. Hence, the shear provides a cosmological
encoding of the Schwarzschild tidal field. In the limit \(T\to0\), curvature invariants diverge, signaling the usual
spacelike singularity at finite proper time.

\textit{Bardeen solution.} We now apply the same perspective to the regular Bardeen black hole
\cite{Bardeen:1968,AyonBeato:2000,Ansoldi:2008}. The static Bardeen metric is in the form of \eqref{spheresym}, with 
\begin{equation}
\begin{split}
 f(r)&=1-\frac{2Mr^2}{(r^2+g^2)^{3/2}},
\end{split}
\label{eq:BardeenMetric}
\end{equation}
where \(M\) is the mass and \(g>0\) is the magnetic monopole charge of the nonlinear
electrodynamics source. This spacetime is regular at \(r=0\), with \(f(0)=1\),
and asymptotically flat, with \(f(\infty)=1\). Moreover,
\begin{equation}
 f'(r)=\frac{2Mr\,(r^2-2g^2)}{(r^2+g^2)^{5/2}},
\label{eq:BardeenPrime}
\end{equation}
so the nontrivial extremum occurs at \(r=\sqrt{2}\,g\), where
\begin{equation}
 f_{\min}=1-\frac{4M}{3\sqrt{3}\,g}.
\label{eq:fmin}
\end{equation}
Hence the geometry has two horizons \(r_-<r_+\) provided
\begin{equation}
 M>\frac{3\sqrt{3}}{4}\,g,
\label{eq:BardeenHorizCond}
\end{equation}
a degenerate horizon at equality, and no horizons below that threshold.~\footnote{A separate issue is the stability of the inner horizon \cite{Poisson:1990eh,Ori:1991zz,Brady:1995ni}. 
As with other Cauchy horizons, blueshift and mass-inflation effects may become 
important.}

In the trapped region between horizons, \(r_-<r<r_+\), one has \(f(r)<0\), so
\(t\) and \(r\) exchange causal character. Defining \(T\equiv r\) and \(\chi\equiv t\)
as before, the metric is written as \eqref{interiorregion},
where now
\begin{equation}
 |f(T)|=\frac{2MT^2}{(T^2+g^2)^{3/2}}-1,
 \qquad r_-<T<r_+.
\label{eq:BardeenAbsF}
\end{equation}
Introducing proper time by $\frac{d\tau}{dT}=|f(T)|^{-1/2},$ we obtain a Kantowski--Sachs form \eqref{eq:genKSform} on this finite nonstatic band, with scale factors
\begin{equation}
\begin{split}
 a(\tau)&=\sqrt{|f(T(\tau))|},
 \qquad
 b(\tau)=T(\tau).
\end{split}
\label{eq:BardeenKS}
\end{equation}
For \(0\le r<r_-\), one has \(f(r)>0\) again, so the coordinate
character flips back once more and the deep interior is a regular static core. Near \(r=0\),
\begin{equation}
 f(r)=1-\frac{2M}{g^3}r^2+\mathcal O(r^4),
\label{eq:BardeenCoreExpansion}
\end{equation}
so the metric approaches a de Sitter-like static patch with effective
cosmological constant
\begin{equation}
 \Lambda_{\rm eff}=\frac{6M}{g^3}.
\label{eq:LambdaEff}
\end{equation}
In particular, all curvature invariants remain bounded; for example, the
Kretschmann scalar approaches the finite value
\begin{equation}
 K(0)=\frac{96M^2}{g^6}.
\label{eq:BardeenK0}
\end{equation}
The resulting sign structure is therefore a static exterior (\(r>r_+\)),
a nonstatic trapped region (\(r_-<r<r_+\)), and a regular static core
(\(0\le r<r_-\)).

\section{Regular black holes and the classical double copy}
\label{sec:staticdc}

The above geometries belong to the standard Kerr--Schild class
\cite{KerrSchild:2009rep,Monteiro:2014cda,Easson:2023iso}. A static,
spherically symmetric, asymptotically flat black-hole solution may be written
on a flat Minkowski background whose line element in Kerr--Schild spherical
coordinates $(u,r,\theta,\varphi)$ is
\begin{equation}
 ds_\eta^2=-du^2+dr^2+r^2d\Omega_2^2,
\end{equation}
as
\begin{equation}
\begin{split}
 g_{\mu\nu}&=\eta_{\mu\nu}+2\phi(r)k_\mu k_\nu,\\
 k_\mu dx^\mu&=du+dr,\\
 k^\mu k_\mu&=0,
\end{split}
\label{eq:KS}
\end{equation}
where $\phi(r)$ is a scalar profile. For the Bardeen geometry,
\begin{equation}
 \phi(r)=\frac{Mr^2}{(r^2+g^2)^{3/2}},
 \qquad
 f(r)=1-2\phi(r),
\label{eq:phi_static}
\end{equation}
so that the horizon condition is simply $f(r)=0$.

Following the classical double-copy prescription
\cite{Monteiro:2014cda,Luna:2015paa,Ridgway:2015fdl,Easson:2020nsbh}, the
Kerr--Schild quantities $(\eta_{\mu\nu},k_\mu,\phi(r))$ map to the Abelian gauge
potential
\begin{equation}
 A=\phi(r)(du+dr),
\label{eq:SingleCopyKS}
\end{equation}
together with a zeroth-copy scalar carrying the same radial profile
$\phi(r)$. In the same Kerr--Schild chart the metric takes the form
\begin{equation}
 ds^2=-(1-2\phi)\,du^2+4\phi\,du\,dr+(1+2\phi)\,dr^2+r^2d\Omega_2^2.
\label{eq:KSMetricoffdiag}
\end{equation}
The usual diagonal static form is recovered by the radial time redefinition
\begin{equation}
 dt=du-\frac{2\phi(r)}{1-2\phi(r)}\,dr,
\label{eq:DiagTimeShift}
\end{equation}
for which
\begin{equation}
 ds^2=-(1-2\phi(r))dt^2+(1-2\phi(r))^{-1}dr^2+r^2d\Omega_2^2.
\label{eq:StaticDiagonalMetric}
\end{equation}

The same gauge field is gauge-equivalent to a
purely electrostatic potential,
\begin{equation}
 A=\phi(r)\,dt+\frac{\phi(r)}{1-2\phi(r)}\,dr
 \;\sim\; \phi(r)\,dt,
\label{eq:SingleCopyDiag}
\end{equation}
because on any region with $f(r)=1-2\phi(r)\neq0$ the radial term is exact:
\begin{equation}
 \frac{\phi(r)}{1-2\phi(r)}\,dr=d\lambda(r),
 \qquad
 \lambda(r)=\int^r ds\,\frac{\phi(s)}{1-2\phi(s)}.
\end{equation}
The corresponding field strength may be written in either chart as
\begin{equation}
 F=-\phi'(r)\,dt\wedge dr=-\phi'(r)\,du\wedge dr.
\label{eq:Fbothcharts}
\end{equation}

In the present matter-supported examples the gauge field is sourced, and the
same is true of the corresponding zeroth-copy scalar. On the flat background,
Maxwell's equation takes the form
\begin{equation}
 \nabla^{(\eta)}_{\nu}F^{\mu\nu}=J^\mu.
\end{equation}
For a static spherically symmetric field, the only nonzero component is the
temporal one. In the Kerr--Schild chart one may write
\begin{equation}
\begin{split}
 J^u(r)&=\frac{1}{r^2}\frac{d}{dr}\left(r^2\phi'(r)\right),\\
 J^r&=J^\theta=J^\varphi=0,
\end{split}
\label{eq:Jugeneral}
\end{equation}
while in the diagonal static chart the same radial profile is denoted by
$J^t(r)=J^u(r)$. For the Bardeen case this source density remains finite at
$r=0$, mirroring the regularity of the gravitational core. This
Kerr--Schild structure is the starting point for the interior analysis that
follows.

\section{Interior observers and the classical single copy}
\label{sec:Bardeen}

We now specialize the static Kerr--Schild data of
Sec.~\ref{sec:staticdc} to a trapped interval. The metric is Eq.~\eqref{eq:StaticDiagonalMetric}, with
\(f(r)=1-2\phi(r)\) as in Eq.~\eqref{eq:phi_static}. On any interval where
\(f(r)<0\), the coordinate change \(T\equiv r\), \(\chi\equiv t\), together
with the proper-time definition \eqref{propertime}, puts the geometry in the
Kantowski--Sachs form \eqref{eq:genKSform}, with scale factors given by
Eq.~\eqref{KSprofilefunctions}. The associated flat-background single copy is given by
Eqs.~\eqref{eq:SingleCopyKS} and \eqref{eq:SingleCopyDiag}, with field
strength \eqref{eq:Fbothcharts}. Because the field strength is invariant under the time redefinition
\eqref{eq:DiagTimeShift}, its restriction to the trapped patch is
\begin{equation}
 F=-\phi'(T)\,d\chi\wedge dT.
\label{eq:GeneralTrappedF}
\end{equation}
A convenient gauge-equivalent potential is
\begin{equation}
 A\sim \phi(T)\,d\chi.
\label{eq:GeneralTrappedA}
\end{equation}
Thus every trapped interval of the static, spherically symmetric
Kerr--Schild class admits a local single-copy presentation on the associated
Kantowski--Sachs patch. Its intrinsic reconstruction and complementary
curved-background interpretation are developed below.

\subsection{Intrinsic reconstruction of trapped interiors}
\label{sec:intrinsic}

The trapped Kerr--Schild class admits an intrinsic characterization in terms of
the Kantowski--Sachs metric alone.

\begin{theorem}[Intrinsic characterization]
\label{thm:KS_characterization}
Consider a Kantowski--Sachs metric
\begin{equation}
 ds^2=-d\tau^2+a(\tau)^2d\chi^2+b(\tau)^2d\Omega_2^2,
\end{equation}
on a connected interval with $a,b>0$ and $\dot b\neq0$.

Then the following are equivalent:
\begin{enumerate}
\item The metric arises locally, up to a constant normalization of the
longitudinal coordinate, from a trapped interval of a static, spherically
symmetric one-function metric
\begin{equation}
\begin{aligned}
ds^2&=-f(r)\,dt^2+\frac{dr^2}{f(r)}+r^2 d\Omega_2^2,\\
f(r)&=1-2\phi(r).
\end{aligned}
\end{equation}
equivalently from a spherically symmetric Kerr--Schild spacetime.

\item The scale factors satisfy the normalization-independent condition
\begin{equation}
 \frac{a(\tau)}{|\dot b(\tau)|}=\text{constant}.
\end{equation}
\end{enumerate}

If $a/|\dot b|=\lambda>0$, the canonical constant normalization
$\widetilde\chi=\lambda\chi$ gives
$\widetilde a=a/\lambda=|\dot b|$. In this canonical normalization,
choosing a time orientation with $\dot b>0$ gives $a=\dot b$, and the parent
geometry is reconstructed by
\begin{equation}
\begin{aligned}
T&=b(\tau), \qquad f(T)=-\dot b(\tau(T))^2=-a(\tau(T))^2,\\
\phi(T)&=\frac{1+\dot b(\tau(T))^2}{2}
       =\frac{1+a(\tau(T))^2}{2}.
\end{aligned}
\end{equation}
\end{theorem}

\begin{proof}
On a trapped interval $f(r)<0$, defining $T=r$ and introducing proper time
via $d\tau=dT/\sqrt{|f(T)|}$ yields the canonical Kantowski--Sachs form with
$a=\sqrt{|f|}$ and $b=T$, hence $a=|\dot b|$. A constant rescaling of the
longitudinal coordinate changes $a$ by the inverse constant factor, so the
normalization-independent statement is $a/|\dot b|=\mathrm{constant}$.

Conversely, suppose $a/|\dot b|=\lambda>0$. With
$\widetilde\chi=\lambda\chi$, the longitudinal scale factor becomes
$\widetilde a=a/\lambda=|\dot b|$. Since $\dot b\neq0$, $T=b(\tau)$ is a
valid local coordinate, and the metric becomes
\[
ds^2=-\frac{dT^2}{\dot b(\tau(T))^2}
      +\dot b(\tau(T))^2d\widetilde\chi^2+T^2d\Omega_2^2,
\]
which is the trapped form of a static, spherically symmetric metric with
$f(T)=-\dot b(\tau(T))^2$. Writing $f=1-2\phi$ gives the stated
reconstruction.
\end{proof}

Both Schwarzschild and Bardeen satisfy the canonical condition $a=\dot b$
in the time orientation used here, and therefore fall within the class
characterized by Theorem~\ref{thm:KS_characterization}.

\paragraph*{Einstein-equation form.}
For the Kantowski--Sachs line element \eqref{eq:genKSform}, define
\[
H_\parallel\equiv \frac{\dot a}{a},
\qquad
H_\perp\equiv \frac{\dot b}{b}.
\]
Einstein's equations then take the form
\begin{equation}\label{eq:EinstenEqs}
\begin{split}
H_\perp^2 + 2 H_\parallel H_\perp + \frac{1}{b^2}
&= 8\pi \rho, \\
2\dot H_\perp + 3H_\perp^2 + \frac{1}{b^2}
&= -8\pi p_\parallel, \\
\dot H_\parallel + \dot H_\perp + H_\parallel^2 + H_\perp^2 + H_\parallel H_\perp
&= -8\pi p_\perp .
\end{split}
\end{equation}

The relation between the one-function static condition and a vanishing radial
null--null stress component, equivalently radial pressure equal to minus the
energy density, was emphasized in \cite{Jacobson:2007}. The result
below is its proper-time Kantowski--Sachs form on the nonconstant-radius
trapped branch relevant here.

\begin{corollary}[Einstein-equation characterization]
\label{cor:parallelEOS}
Let the Kantowski--Sachs metric satisfy Einstein's equations with anisotropic
source
\[
T^\mu{}_\nu=\mathrm{diag}(-\rho,p_\parallel,p_\perp,p_\perp)
\]
on a connected interval with $a,b>0$ and $\dot b\neq0$. Then
\begin{equation}
 p_\parallel=-\rho
\end{equation}
is equivalent to
\begin{equation}
 \frac{a}{|\dot b|}=\text{constant},
\end{equation}
and hence, after the canonical constant normalization of $\chi$, to
$a=|\dot b|$.
\end{corollary}

\begin{proof}
Combining the first two equations in Eq.~\eqref{eq:EinstenEqs} gives
\begin{equation}
 8\pi(\rho+p_\parallel)
 =\frac{2}{b}\left(\frac{\dot a}{a}\dot b-\ddot b\right)
 =-\frac{2a}{b}\frac{d}{d\tau}\left(\frac{\dot b}{a}\right).
\end{equation}
Thus $p_\parallel=-\rho$ implies
\begin{equation}
 \dot b=C\,a
\end{equation}
for a constant $C$. $C=0$ corresponds to the
constant-radius branch, $\dot b=0$; that branch is excluded by the hypotheses
of Theorem~\ref{thm:KS_characterization} and of the present corollary. On the
nonconstant-radius branch $C\neq0$, so $a/|\dot b|=1/|C|$ is constant. The
canonical rescaling
\begin{equation}
 \widetilde\chi=\frac{\chi}{|C|},
 \qquad
 \widetilde a=|C|a=|\dot b|
\end{equation}
then gives the normalization used in the reconstruction theorem.

Conversely, if $a/|\dot b|$ is constant on the connected interval, the sign of
$\dot b$ is fixed and $\dot b/a$ is constant. The derivative above therefore
vanishes, giving $\rho+p_\parallel=0$ and hence $p_\parallel=-\rho$.
\end{proof}

Within the class characterized by Theorem~\ref{thm:KS_characterization}, the
normalization-independent relation $a/|\dot b|=\mathrm{constant}$ is therefore
equivalent to $p_\parallel=-\rho$, in agreement with
Corollary~\ref{cor:parallelEOS}. In the canonical normalization this becomes
$a=|\dot b|$ (and $a=\dot b$ in the orientation used below). Thus the trapped
Kerr--Schild interiors picked out by our classification are precisely those
satisfying the longitudinal relation $p_\parallel=-\rho$, with Schwarzschild
realizing this trivially through $\rho=p_\parallel=0$.

For Schwarzschild, using \eqref{eq:HubbleSch}, the left-hand sides vanish
identically, consistent with \(\rho=p_\parallel=p_\perp=0\). For Bardeen,
the same longitudinal condition is compatible with the regular de Sitter-like
core \eqref{eq:BardeenCoreExpansion}.

\begin{theorem}[Intrinsic reconstruction of the single copy]
\label{thm:reconstruction}
On the class characterized by Theorem~\ref{thm:KS_characterization},
the Kerr--Schild scalar and the associated Abelian single copy are
uniquely determined by the Kantowski--Sachs quantities within the spherical
Kerr--Schild prescription, once the canonical constant normalization of the
longitudinal coordinate and an orientation are fixed.

Normalization independently,
\begin{equation}
 \phi(T)=\frac{1+\dot b(\tau(T))^2}{2}.
\end{equation}
In the canonical normalization $a=|\dot b|$, this is equivalently
$\phi=(1+a^2)/2$. Choosing the orientation $\dot b=a>0$ and the orthonormal
frame \(e^{\hat 0}=d\tau\), \(e^{\hat 1}=a(\tau)\,d\chi\), the reconstructed
field may be written as
\begin{equation}\label{eq:Eabrelations}
\begin{aligned}
 F&=-\phi'(T)\,d\chi\wedge dT
   =E_{\hat 1}(\tau)\,e^{\hat 0}\wedge e^{\hat 1},\\
 E_{\hat 1}(\tau)
   &=\phi'(T(\tau))
    =\dot a(\tau)
    =\ddot b(\tau),
\end{aligned}
\end{equation}
up to the choice of orientation and sign convention.
\end{theorem}

\begin{proof}
The normalization-independent reconstruction follows from
$f(T)=-\dot b^2$ and $f=1-2\phi$. In the canonical normalization and chosen
orientation, $a=\dot b=dT/d\tau$, which gives the remaining displayed
relations.
\end{proof}

The reconstructed one-form admits a complementary interpretation as a
Maxwell field on the full Kantowski--Sachs geometry. The inheritance of
suitably aligned electromagnetic fields between flat and Kerr--Schild
backgrounds has a prior structural precedent \cite{Kassandrov:2011}; here we use it in the setting of the intrinsic trapped-interior reconstruction
and its homogeneous time-dependent interpretation. Regard
\begin{equation}
 A_{\rm curv}
 =
 \phi\bigl(b(\tau)\bigr)\,d\chi
 \label{eq:CurvedKSPotential}
\end{equation}
as a Maxwell potential on
\begin{equation}
 ds^2=-d\tau^2+a(\tau)^2d\chi^2+b(\tau)^2d\Omega_2^2,
 \qquad
 \dot b=a .
 \label{eq:CurvedKSMetricForMaxwell}
\end{equation}
Then
\begin{equation}
 F_{\tau\chi}
 =
 \partial_\tau\phi\bigl(b(\tau)\bigr)
 =
 a(\tau)\phi'\bigl(b(\tau)\bigr),
 \label{eq:CurvedKSFieldComponent}
\end{equation}
and Maxwell's equation on the Kantowski--Sachs background gives
\begin{equation}
 J^\chi_{\rm curv}
 \equiv
 \nabla^{(g)}_{\nu}F^{\chi\nu}
 =
 \frac{1}{b^2}\frac{d}{db}
 \left[b^2\phi'(b)\right].
 \label{eq:CurvedKSCurrent}
\end{equation}
Thus the same source profile that appears as the static flat-background
charge density in Eq.~\eqref{eq:Jugeneral} appears on the full interior
geometry as a longitudinal current. In this curved-background realization the
field is homogeneous and time dependent in the comoving frame, with
\begin{equation}
 E_{\hat 1}(\tau)=\phi'\bigl(b(\tau)\bigr),
 \label{eq:CurvedKSElectricField}
\end{equation}
in agreement with Eq.~\eqref{eq:Eabrelations}.

If the Kerr--Schild profile has Schwarzschild-type behavior with asymptotic
control through two derivatives,
\begin{equation}
\begin{aligned}
 \phi(T)&=\frac{M}{T}+o(T^{-1}),\\
 \phi'(T)&=-\frac{M}{T^2}+o(T^{-2}),\\
 \phi''(T)&=\frac{2M}{T^3}+o(T^{-3}),
 \qquad (T\to0),
\end{aligned}
\end{equation}
then
\begin{equation}
 F\sim \frac{M}{T^2}\,d\chi\wedge dT,
 \qquad
 K\sim \frac{48M^2}{T^6},
\label{eq:SingularAsymptotics}
\end{equation}
so both the gauge-side field strength and the gravitational curvature diverge.
If instead the Kerr--Schild profile extends smoothly to a de Sitter-like core,
\begin{equation}
 \phi(r)=\frac{\Lambda_{\rm eff}}{6}r^2+\mathcal O(r^4)
 \qquad (r\to 0),
\end{equation}
then
\begin{equation}
 F=-\frac{\Lambda_{\rm eff}}{3}r\,dt\wedge dr+\mathcal O(r^3),
 \qquad
 K\to \frac{8}{3}\Lambda_{\rm eff}^2,
\label{eq:RegularAsymptotics}
\end{equation}
so the single-copy field extends smoothly to the center and vanishes there,
while the gravitational core remains curvature-regular \cite{Easson:2020nsbh}.

For Schwarzschild and Bardeen, in either chart the corresponding field strength is purely electric,
\begin{equation}
\begin{split}
 F_{tr}&=-\phi'(r),\\
F_{\mu\nu}F^{\mu\nu}&=-2[\phi'(r)]^2,\\ F_{\mu\nu}\tilde F^{\mu\nu}&=0.
\end{split}
\label{eq:staticinvariants}
\end{equation}

In the example formulas below we denote the temporal component in the
diagonal static chart by $J^t$; it has the same radial profile as the
Kerr--Schild-chart component $J^u$ in Eq.~\eqref{eq:Jugeneral}.

For Schwarzschild,
\begin{equation}
 \phi_{\rm Schw}(r)=\frac{M}{r},
 \qquad
 f(r)=1-\frac{2M}{r},
\end{equation}
so that
\begin{equation}
 F_{tr}^{\rm Schw}=\frac{M}{r^2},
 \qquad
 J^t_{\rm Schw}(r)=0\qquad(r>0),
\label{eq:SchwJ}
\end{equation}
with only the familiar distributional point source at $r=0$. 

For Bardeen,
\begin{equation}
 \phi_{\rm B}(r)=\frac{Mr^2}{(r^2+g^2)^{3/2}},
 \qquad
 f(r)=1-\frac{2Mr^2}{(r^2+g^2)^{3/2}},
\label{eq:phiBardeen}
\end{equation}
and
\begin{equation}
 \phi_{\rm B}'(r)=\frac{Mr(2g^2-r^2)}{(r^2+g^2)^{5/2}},
\label{eq:phiBprime}
\end{equation}
so that
\begin{equation}
 F_{tr}^{\rm B}(r)=-\phi_{\rm B}'(r)
 =-\frac{Mr(2g^2-r^2)}{(r^2+g^2)^{5/2}}.
\label{eq:FBardeen}
\end{equation}
The source density in the diagonal static chart is
\begin{equation}
\begin{split}
 J^t_{\rm B}(r)&=\frac{1}{r^2}\frac{d}{dr}\left(r^2\phi_{\rm B}'(r)\right)\\
 &=\frac{3Mg^2(2g^2-3r^2)}{(r^2+g^2)^{7/2}},
\end{split}
\label{eq:JBardeen}
\end{equation}
which is smooth for all $r$ and finite at the core $r=0$.
In the complementary curved-background realization,
Eq.~\eqref{eq:CurvedKSCurrent} gives the same source profile. Under the natural
Kerr--Schild identification, the contravariant current can in fact be
identified directly in the two descriptions,
\begin{equation}
\begin{aligned}
J_{(\eta)}=J_{(g)}
 &= j(r)\,\partial_u
  = j(r)\,\partial_t
  = j(b)\,\partial_\chi, \\
j(r)
 &= \frac{1}{r^2}\frac{d}{dr}
    \left[r^2\phi'(r)\right].
\end{aligned}
\end{equation}
What changes is its metric-dependent causal character and physical
interpretation: $\eta(J,J)=-j^2$, whereas $g(J,J)=-f\,j^2$, so on a trapped
interval $f<0$ the same nonzero vector is spacelike with respect to the full
geometry. For Bardeen, with $b=T(\tau)$,
\begin{equation}
 J^\chi_{{\rm curv},\,{\rm B}}(\tau)
 =
 \frac{3Mg^2\bigl(2g^2-3T(\tau)^2\bigr)}
 {\bigl(T(\tau)^2+g^2\bigr)^{7/2}},
\label{eq:JBardeenCurved}
\end{equation}
which remains finite throughout the trapped region. Its radial profile has a
zero at $T^2=2g^2/3$, but that zero lies inside the open trapped
Kantowski--Sachs band only for
$M/g>5\sqrt{15}/12$; at equality it lies on the inner horizon, while below
that threshold it lies in the static core. For Schwarzschild the corresponding
current vanishes away from the singular endpoint. Thus the two descriptions share the same contravariant current profile under the natural Kerr--Schild identification, while the background metric determines its causal character and physical interpretation.

The flat-background single copy is therefore a sourced Abelian Maxwell field with finite net charge fixed by the asymptotic \(1/r\) behavior of the scalar profile.

The signed field components in the two realizations differ by the convention
already visible in Eqs.~\eqref{eq:CurvedKSElectricField} and
\eqref{eq:staticinvariants}. On the chosen canonical branch $\dot b=a>0$,
\begin{equation}
 E_{\hat1}=\phi'(T)=-F_{tr},
\end{equation}
where $E_{\hat1}$ is the coefficient in the curved orthonormal frame and
$F_{tr}$ is the signed flat-background component in the diagonal static chart.
The quantities below use the latter convention and simply evaluate that radial
profile along the Kantowski--Sachs evolution $r=T(\tau)$. For Schwarzschild,
\begin{equation}
 E_{\rm Schw}(\tau)=F_{tr}^{\rm Schw}\big|_{r=T(\tau)}
 =\frac{M}{T(\tau)^2},
\label{eq:ESchw}
\end{equation}
which diverges as the singularity $T\to0$ is approached. For Bardeen,
\begin{equation}
\begin{split}
 E_{\rm B}(\tau)&=F_{tr}^{\rm B}\big|_{r=T(\tau)}\\
 &=-\frac{M T(\tau)\bigl(2g^2-T(\tau)^2\bigr)}{\bigl(T(\tau)^2+g^2\bigr)^{5/2}},
\end{split}
\label{eq:EBardeenTau}
\end{equation}
which remains finite throughout the trapped interval $r_-<r<r_+$.

Here $\tau$ is specifically the proper time of the fixed-$\chi$
Kantowski--Sachs comoving congruence. Since $\chi=t$, that congruence has zero
conserved Killing energy and does not continue as the same timelike congruence
into the static region $f>0$. The underlying single-copy field
$F_{tr}^{\rm B}(r)$ itself extends smoothly across the inner horizon into the
regular static core and vanishes at $r=0$. Generic radial infallers with
nonzero Killing energy may cross the inner horizon and continue to sample this
same regular field, but their proper time is a distinct parameter (say $s$),
not the Kantowski--Sachs time $\tau$; static observers in the core measure the
same regular electrostatic profile at fixed $r$.

The electromagnetic stress tensor is
\begin{equation}
 T_{\mu\nu}^{(\mathrm{EM})}
 =F_{\mu\alpha}F_{\nu}{}^{\alpha}-\frac14\eta_{\mu\nu}F_{\alpha\beta}F^{\alpha\beta},
\label{eq:TEM}
\end{equation}
and the electric energy density is
\begin{equation}
 \rho_{\mathrm{EM}}(r)=\frac12[\phi'(r)]^2.
\label{eq:rhoEMgeneral}
\end{equation}
Hence
\begin{equation}
 \rho_{\mathrm{EM}}^{\rm Schw}(r)=\frac{M^2}{2r^4},
\label{eq:rhoEMSchw}
\end{equation}
which diverges at $r=0$, while for Bardeen
\begin{equation}
 \rho_{\mathrm{EM}}^{\rm B}(r)=\frac{M^2r^2(2g^2-r^2)^2}{2(r^2+g^2)^5},
\label{eq:rhoEMB}
\end{equation}
which is finite for all $r$ and in particular satisfies
\begin{equation}
 \rho_{\mathrm{EM}}^{\rm B}(0)=0.
\end{equation}

\begin{figure}[t]
\centering
\includegraphics[width=1\linewidth]{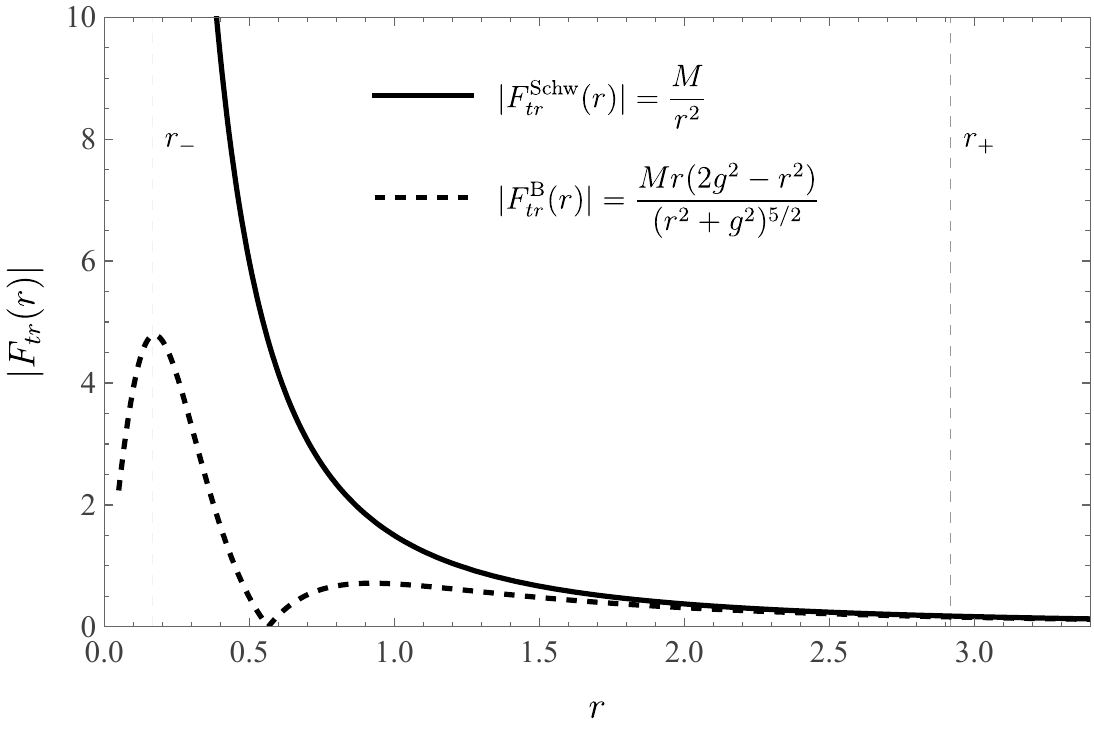}
\caption{Comparison of the electric single-copy profiles for Schwarzschild and
Bardeen, plotting $|F_{tr}(r)|=|\phi'(r)|$ from \eqref{eq:SchwJ} and \eqref{eq:FBardeen}. For Schwarzschild,
$|F_{tr}^{\rm Schw}(r)|=M/r^2$, while for Bardeen,
$|F_{tr}^{\rm B}(r)|=\left|Mr(2g^2-r^2)/(r^2+g^2)^{5/2}\right|$.
The curves are shown for $M=1.5$ in both cases and $g=0.4$ for the Bardeen
solution. For these parameters the Bardeen horizon locations are
$r_-\approx0.16$ and $r_+\approx2.9$, indicated by dashed vertical lines.
The Schwarzschild profile diverges as $r\to0$, whereas the Bardeen profile
remains finite on the trapped interval, extends smoothly across the inner
horizon, and vanishes at the regular center $r=0$.}
\label{fig:singlecopyprofiles}
\end{figure}

Thus the standard Kerr--Schild single copy remains static on
\(\eta_{\mu\nu}\), while the complementary Maxwell field on
\(g_{\mu\nu}\) is homogeneous and time dependent. The two descriptions share
the same reconstructed profile but differ in their background geometry and
source interpretation.

\section{Regularity and Energy Conditions Across the Double Copy}
\label{sec:energy}

The Bardeen spacetime is supported by an anisotropic effective matter sector
with stress tensor \cite{AyonBeato:2000,Easson:2020nsbh}
\begin{equation}
 T^\mu{}_{\nu}=\mathrm{diag}\bigl(-\rho(r),p_r(r),p_T(r),p_T(r)\bigr),
\label{eq:Tgrav}
\end{equation}
where
\begin{equation}
 \rho(r)=\frac{3Mg^2}{4\pi\bigl(r^2+g^2\bigr)^{5/2}},
 \qquad
 p_r(r)=-\rho(r),
\label{eq:rhopr}
\end{equation}
and
\begin{equation}
 p_T(r)=\frac{3Mg^2\bigl(3r^2-2g^2\bigr)}{8\pi\bigl(r^2+g^2\bigr)^{7/2}}.
\label{eq:pT}
\end{equation}
The relation \(p_r=-\rho\) shows that the source is vacuum-like in the radial direction. Equivalently, wherever \(\rho\neq 0\), the
effective longitudinal equation-of-state parameter is
\begin{equation}
 w_{\parallel}\equiv \frac{p_r}{\rho}=-1.
\end{equation}
The source nevertheless remains anisotropic away from the center because
\(p_T\neq p_r\). At the regular center \(r=0\), however, the anisotropy
disappears and the core becomes de Sitter-like and isotropic,
\begin{equation}
 p_r(0)=p_T(0)=-\rho(0).
\end{equation}

Using the standard pointwise energy-condition definitions
\cite{Wald:1984gr,HawkingEllis:1973}, the null energy condition is controlled by
\begin{equation}
 \rho+p_r=0,
\end{equation}
and
\begin{equation}
 \rho+p_T=\frac{15Mg^2r^2}{8\pi\bigl(r^2+g^2\bigr)^{7/2}}\ge0.
\label{eq:NECgrav}
\end{equation}
Thus the null energy condition is satisfied everywhere. More precisely, the
radial null energy condition is saturated identically, while the tangential
null energy condition saturates only at \(r=0\). The strong energy condition
depends on
\begin{equation}
 \rho+p_r+2p_T=2p_T
 =\frac{3Mg^2\bigl(3r^2-2g^2\bigr)}{4\pi\bigl(r^2+g^2\bigr)^{7/2}},
\label{eq:SECgrav}
\end{equation}
which is negative for
\begin{equation}
 r^2<\frac{2}{3}g^2.
\end{equation}
The location of this domain relative to the trapped region depends on \(M/g\).
The surface \(r=\sqrt{2/3}\,g\) coincides with the inner horizon at
\begin{equation}
 \frac{M}{g}=\frac{5\sqrt{15}}{12}.
\label{eq:SECInnerHorizonThreshold}
\end{equation}

Below this value the SEC-violating region is confined to the static core,
whereas above it the violation extends into a finite portion of the trapped
Kantowski--Sachs band.
Thus the Bardeen geometry violates the strong energy condition in a compact
central region while preserving the null energy condition everywhere. In this sense, the regularization is controlled by
the transverse sector: the radial null energy condition is saturated
identically, while strong-energy-condition violation occurs only when the
tangential pressure becomes sufficiently negative near the center.

For completeness, the Bardeen source satisfies the weak energy condition
everywhere, since \(\rho(r)\ge0\), \(\rho+p_r=0\), and \(\rho+p_T\ge0\). By
contrast, the dominant energy condition fails for \(r>2g\), where
\(|p_T|>\rho\); in particular, this inequality persists in the asymptotic
region. For a more in-depth discussion of energy conditions in nonsingular
black hole spacetimes see
\cite{Zaslavskii:2010qz,Maeda:2021jdc,Davies:2024ysj}.

The gauge-side stress tensor and energy-condition statements below refer to
the standard Maxwell field on the flat background
\(\eta_{\mu\nu}\).
On the gauge side, the field invariants are already given in
Eq.~\eqref{eq:staticinvariants}; the single copy is therefore a purely
electric Maxwell field. Its stress tensor \eqref{eq:TEM} obeys the usual
pointwise classical energy conditions. In particular, for any null vector
$\ell^\mu$,
\begin{equation}
 T_{\mu\nu}^{(\mathrm{EM})}\,\ell^\mu\ell^\nu\ge0.
\end{equation}
If $|\phi'(r)|$ diverges, as in the Schwarzschild case, the gauge-field energy
density also diverges. If $|\phi'(r)|$ remains finite, as in the Bardeen case,
the gauge field remains regular everywhere.

\section{Horizons from the single copy}
\label{sec:horizons}

It is commonly noted that classical double-copy mappings are local and appear
blind to global causal structure, such as the presence of horizons. However,
recent work by Chawla and Keeler \cite{ChawlaKeeler:2023} demonstrated that,
for Kerr--Schild spacetimes, certain local horizon diagnostics---specifically
the vanishing of the expansion of outgoing null congruences---can be expressed
purely in terms of the Kerr--Schild scalar and null vector, i.e. single-copy
quantities. This local
criterion identifies marginally trapped surfaces, whose locations coincide
with the usual Killing horizons.

For the horizon diagnostic we use the future-directed ingoing sign choice
consistent with Ref.~\cite{ChawlaKeeler:2023},
\begin{equation}
 k_\mu dx^\mu=-du-dr.
\end{equation}
This is an overall sign reversal of the Kerr--Schild covector used above, so
$k_\mu k_\nu$ and hence the metric are unchanged. With the corresponding
outgoing null normal, the expansion formula below gives the same marginality
condition $\Phi(r)=1$.

We now apply this method to the Bardeen black hole to show that the full
horizon phase structure---existence, merger, and disappearance of
horizons---can be recovered from gauge-side quantities alone. For any
Kerr--Schild metric \eqref{eq:ksmet},
the expansion of the outgoing null normal $l^\mu$ is given by
\begin{equation}
 \theta(l)=\theta^{(0)}(l)+\frac{\Phi}{2}\theta^{(0)}(k),
 \qquad
 \Phi\equiv2\phi.
\end{equation}
For spherically symmetric slices in flat space,
$\theta^{(0)}(k)=-2/r$ and $\theta^{(0)}(l)=1/r$, yielding
\begin{equation}
 \theta(l)=\frac{1-\Phi(r)}{r}.
\end{equation}
The condition for a marginally trapped surface is therefore
\begin{equation}
 \Phi(r)=1.
\end{equation}

In Schwarzschild, this recovers the usual horizon at $r=2M$, where
$\phi(r)=M/r$. For the Bardeen solution,
\begin{equation}
 \Phi(r)=\frac{2Mr^2}{(r^2+g^2)^{3/2}},
\end{equation}
the horizon condition becomes
\begin{equation}
 2Mr^2=(r^2+g^2)^{3/2}.
\end{equation}
Introducing dimensionless variables $x=r/g$ and $\mu=M/g$, the equation reads
\begin{equation}
 2\mu x^2=(x^2+1)^{3/2}.
\end{equation}
This equation has
\begin{itemize}
\item two real roots for $\mu>\mu_{\rm crit}=3\sqrt{3}/4$,
\item a degenerate root at $x=\sqrt{2}$ for $\mu=\mu_{\rm crit}$,
\item no real roots for $\mu<\mu_{\rm crit}$.
\end{itemize}

\begin{figure}[H]
\centering
\includegraphics[width=1.0 \linewidth]{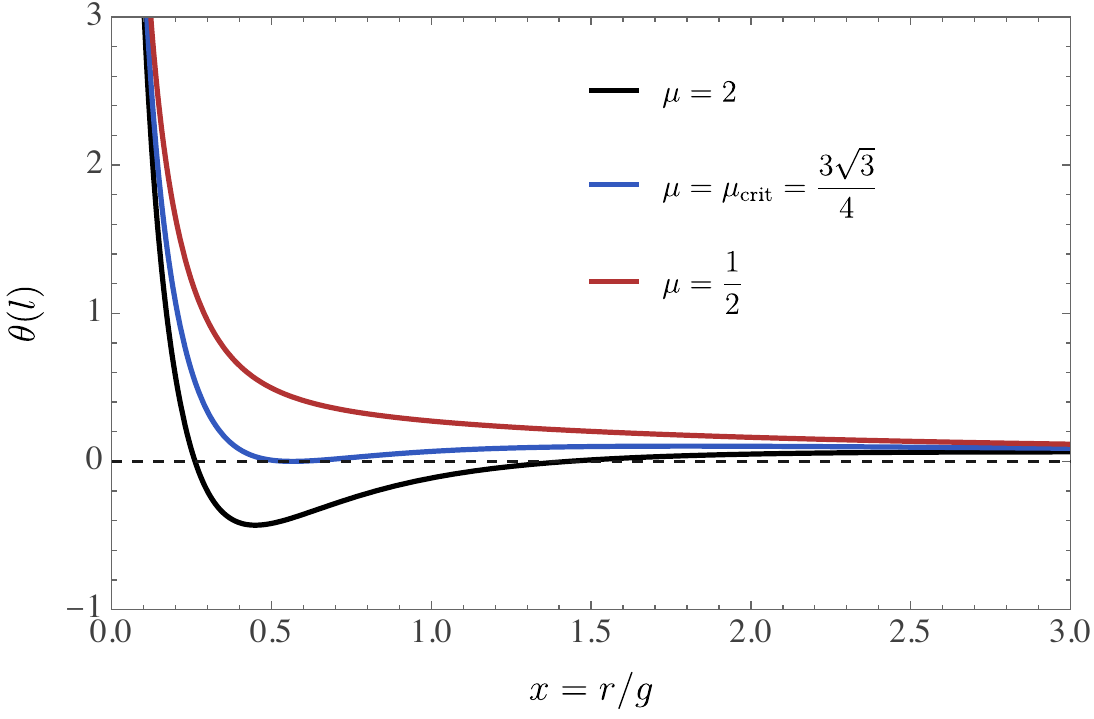}
\caption{$\theta(l)=(1-\Phi(r))/r$ for the Bardeen black hole with $g=0.4$.
Solid: $\mu=2$ (two horizons); dashed:
$\mu=\mu_{\mathrm{crit}}=3\sqrt{3}/4$ (degenerate horizon); dotted: $\mu=1/2$
(no horizons).}
\label{fig:theta}
\end{figure}

Figure~\ref{fig:theta} illustrates this behavior by plotting $\theta(l)$ for
representative values of $\mu$. We stress that this result is derived without
reference to the full gravitational metric: only the scalar profile,
through $\Phi(r)=2\phi(r)$, and the flat-space expansion quantities are needed. The local trapping-horizon structure is encoded in the single-copy data,
even for matter-supported regular black holes.

Unlike the Schwarzschild case, where the Kerr--Schild scalar $\phi(r)$ is
monotonic and the single-copy gauge field shows no qualitative feature at the
horizon, the regular Bardeen black hole exhibits a richer structure. The
nonmonotonic behavior of $\phi(r)$ allows the single-copy condition
$\Phi(r)=1$ to admit two, one, or no real roots, corresponding precisely to
the non-extremal, extremal, and horizonless phases of the spacetime. Thus the
zeros of $\theta(l)$ track the appearance and merger of gravitational
horizons even though the gauge field itself remains smooth and satisfies the
classical energy conditions. The local horizon phase structure, including the
critical extremal transition, is therefore encoded algebraically in the
single-copy scalar. More cautiously, this suggests that appropriate
single-copy charge profiles may provide a useful guide to the horizon
structure of matter-supported gravitational solutions. To our knowledge, the
Bardeen solution provides the first explicit regular multi-horizon
Kerr--Schild example to which the Chawla--Keeler horizon criterion has been
applied, complementing the regular-black-hole double-copy analysis of
Ref.~\cite{Easson:2020nsbh}.

\section{Conclusions}\label{sec:concl}

Black-hole interiors provide a natural setting in which the classical double copy becomes explicitly time dependent. We have shown that trapped intervals of static, spherically symmetric Kerr--Schild geometries admit an exact double-copy description in their Kantowski--Sachs form, with the gauge quantities reconstructible directly from the interior cosmological geometry. The relevant Kantowski--Sachs patches form an intrinsic class characterized by the relation $a/|\dot b|=\mathrm{constant}$, equivalently by the longitudinal condition $p_\parallel=-\rho$. Within this class, the Kerr--Schild scalar and single-copy field are uniquely reconstructible without prior knowledge of an exterior black-hole geometry. The same reconstructed one-form admits both the standard stationary Maxwell realization on the flat Kerr--Schild background and a complementary homogeneous, time-dependent realization on the full Kantowski--Sachs spacetime.

In Schwarzschild, the Coulomb single-copy field diverges along the interior evolution, mirroring the singular Kantowski--Sachs focusing. In
Bardeen, the trapped region is finite, the deep interior is a regular static
core rather than a second Kantowski--Sachs phase, and the inherited
single-copy field remains finite and extends smoothly to the center. The
Bardeen geometry violates the strong energy condition in a compact central
region while preserving the null energy condition, whereas the corresponding
Maxwell field remains regular and satisfies the standard pointwise classical
energy conditions. The same single-copy scalar encodes the non-extremal,
extremal, and horizonless phases of the Bardeen geometry.

Trapped Kerr--Schild interiors therefore tie together the geometry, effective matter content, and single-copy field through intrinsic local relations. On the flat background the standard single copy remains stationary, whereas on the Kantowski--Sachs geometry the corresponding Maxwell realization is homogeneous and time dependent. Natural extensions include perturbations, rotating trapped regions, and more general matter-supported Kerr--Schild geometries.

\acknowledgments
It is our pleasure to thank G.~Herczeg, C.~Keeler and M.~Pezzelle for valuable discussions. DAE is supported in part by the U.S. Department of Energy,
Office of High Energy Physics, under Award Number DE-SC0019470.

\bibliography{refs}

\end{document}